%% file: arxiv_main.tex
\pdfoutput=1

\documentclass[11pt]{article}

\usepackage[margin=1in]{geometry}
\usepackage[english]{babel}
\usepackage[compact]{titlesec}
\usepackage[T1]{fontenc}

\usepackage{authblk} 

\input{macros}

\title{Distribution of Chores with Information Asymmetry}

\author[1]{Hadi Hosseini}
\author[2]{Joshua Kavner}
\author[3]{Tomasz Wąs}
\author[2]{Lirong Xia}
\affil[1]{Pennsylvania State University}
\affil[2]{Rensselaer Polytechnic Institute}
\affil[3]{University of Oxford}

\begin{document}

\maketitle

\begin{abstract}
A well-regarded fairness notion when dividing indivisible chores is envy-freeness up to one item (\EF{1}), which requires that pairwise envy can be eliminated by the removal of a single item. While an \EF{1} and Pareto optimal (\PO{}) allocation of goods can always be found via  well-known algorithms, even the existence of such solutions for chores remains open, to date. We take an \textit{epistemic} approach utilizing information asymmetry by introducing \textit{dubious chores}--items that inflict no cost on receiving agents but are perceived costly by others. On a technical level, dubious chores provide a more fine-grained approximation of envy-freeness than \EF{1}. We show that finding allocations with minimal number of dubious chores is computationally hard. Nonetheless, we prove the existence of envy-free and fractional \PO{} allocations for $n$ agents with only $2n-2$ dubious chores and strengthen it to $n-1$ dubious chores in four special classes of valuations. Our experimental analysis demonstrates that often only a few dubious chores are needed to achieve envy-freeness. 
\end{abstract}

\section{Introduction}

The fair allocation of resources and tasks is a fundamental role of any economy.
It has garnered attention of diverse communities spanning computer science, artificial intelligence, political science and economics due to its broad applicability in healthcare, charitable donations, waste management, and task allocation \citep{Aleksandrov15:Online,B11combinatorial,budish2017course,goldman2015spliddit,parkes2015beyond,pathak2021fair}. 
Traditionally, this field is concerned with allocating \textit{indivisible} items that are positively-valued by agents (i.e., \textit{goods}). However, many practical decisions distribute indivisible negatively-valued tasks (i.e., \textit{chores}) too.

The allocation of chores is fundamentally different from its goods counterpart since chores must be fully allocated, whereas goods may be disposed of at no cost. Moreover, algorithmic techniques and axiomatic approaches developed for fair allocation of goods do not immediately translate to this setting.
Thus, in recent years a large body of work has focused on investigating fairness axioms and algorithmic techniques
specifically for allocation of chores
\citep{Aziz_Rauchecker_Schryen_Walsh_2017,bogomolnaia2019dividing,chaudhury2020dividing,Freeman2020:Equitable,hosseini2022ordinal}.
The canonical fairness notion in the literature is envy-freeness (\EF{}), an intrapersonal property which requires that each agent weakly prefers their own bundle to any other \citep{foley1966resource}. However, \EF{} allocations may not exist and determining their existence is computationally intractable, motivating a number of relaxations.

One well-studied relaxation,
\textit{envy-freeness up to one item} (\EF{1}), 
requires that any pairwise envy between the agents can be eliminated by the counterfactual removal of a single item (a good from the envied agent or a chore from the envious agent) \citep{Lipton04:Approximately,B11combinatorial}.
An \EF{1} allocation of goods always exists and can be computed efficiently along with economic efficiency notions such as Pareto optimality (\PO{}) \citep{barman2018finding,caragiannis2019envy}.
In contrast, for chores, not only computing an \EF{1} and \PO{} allocation is unknown, but even the existence of such allocations remains open to date. (\EF{1} allocations without \PO{} can be computed in polynomial time \citep{aziz2022fair,BSV21approximate}.)
This has led several works to study \EF{1} and \PO{} allocations under restricted domains \citep{aziz2022fair,ebadian22:fairly,garg2023improving,hosseini2022fairly}.

We take a different, \textit{epistemic}, approach
that utilizes information asymmetry. Rather than require counterfactual reasoning to achieve fairness, as with \EF{1}, epistemic fair division considers allocations
that are partially hidden \citep{ABC+18knowledge}
or for which information is withheld about 
the goods \citep{Hosseini2020:Fair,bliznets2023fair}.
To this end, we 
differentiate between the allocation of chores that agents are informed about and the allocation they actually receive. The difference is an over-representation about which tasks agents complete, above-and-beyond those they are actually assigned. That is, agents must only complete a subset of their assigned chores, which may contain duplicates. 

We represent this technically through the introduction and allotment of \emph{dubious chores}, copies of the original chores that bear no cost for agents that receive them, but are seen as costly by others.
This models 
settings where agents do not have
direct means of communication and cannot verify the exact costs incurred by each agent, such as with distributed computing of high-complexity problems or decentralized training of large language models.
Formally, we propose a fairness notion of
\textit{envy-freeness up to $k$ dubious chores} (\DEF{k}).
Consider $n$ agents, $m$ chores,
and an allocation $A = (A_1, \ldots, A_n)$.
We introduce up to $k$ \emph{dubious} copies of the original chores, representing additional tasks agents appear to be completing, that are distributed via the \emph{dubious allocation} $A^D = (A^D_1,\dots, A^D_n)$. 
If every agent $i$ prefers its own allocation $A_i$ to the perceived allocation $A_h \cup A^D_h$ (combined real and dubious) of each other agent $h \neq i$, then we say that allocation $A$ is \DEF{k}.

This approach differs from related work about \EF{1}, duplicating chores, and chores with subsidies.
First, agents in our approach measure whether an allocation is fair
based on their information,
which is subjective and may differ across agents.
This contrasts \EF{1} which requires agents to counterfactually reason about 
other agents' bundles. Moreover, \DEF{k} offers a more fine-grained approximation of \EF{} than \EF{1}. Whereas the latter recognizes any allocation \emph{close to} an envy-free allocation, \DEF{k}
precisely identifies a trade-off between fairness and transparency: a \DEF{0} allocation is necessarily \EF{}, while at most $k$ units of transparency must be compromised to make a \DEF{k} allocation envy-free for $k>0$.

Second, 
\citet{akrami2023fair} recently introduced real copies of chores (i.e., ``\emph{surplus}'') into the original fair division instance. Our approach differs on a conceptual level in that
duplicates introduce \textit{real} additional cost on the receiving agents, reducing overall welfare, and more duplicated chores than dubious ones may be needed to eliminate envy. 
%
Third, our method differs from approaches of eliminating envy by subsidizing agents with positively-valued money 
\citep{brustle2020one,caragiannis2021computing,halpern2019fair}. Whereas these approaches change real allocations by introducing unit-valued goods to achieve \EF{} fairness, agents in our approach perceive their allocation as fair according to their visible information.

\subsection{Our Contributions}

We propose a novel fairness notion, \DEF{k}, that utilizes information asymmetry with dubious chores, items that are perceived as costly but inflict no actual cost on the receiving agent. 
Conceptually, dubious chores provide a natural approach 
when no envy-free solution exists.
Technically,
\DEF{k} provides a more fine-grained approximation of envy-freeness than \EF{1}
which enables progress towards addressing open problems in fair allocation of indivisible chores, such as the existence and computation of \EF{1} and \PO{}.
As such, we make the following technical contributions for agents with additive valuations:
\begin{itemize}
    \item We show the hardness of finding the minimal $k$ for which \DEF{k} exists and 
    its constant approximation (\cref{cor:no_approx:defk}). In contrast, we obtain that \DEF{(n-1)} allocations always exist and can be computed efficiently (\cref{thm:round-robin:defn-1}).

    \item 
    We achieve both fairness and efficiency by showing that there always exists an allocation satisfying \DEF{(2n-2)} and fractional Pareto optimality (\fPO{}), a more demanding efficiency requirement than \PO{} (\cref{thm:def2n-2+po}).
    Furthermore, we strengthen this result showing that \DEF{(n-1)} and \fPO{} allocations can be efficiently computed in four special cases: when agents' valuations are identical (\cref{thm:defn-1+po:identical}), binary (\cref{thm:defn-1+po:binary}), or bivalued (\cref{thm:defn-1+po:bivalued}), or upon restricting chores to be two-typed (\cref{thm:defn-1+po:2types}).
    Under binary valuations our algorithm also guarantees \emph{envy-freeness up to any chore} (\EFX{})~\cite{aziz2022fair,Caragiannis19:unreasonable}, a strengthening of \EF{1} under which each pairwise envy can be eliminated by the removal of \emph{any} negatively-valued chore of the envious agent.

    \item Our empirical study demonstrates that \texttt{Round Robin} produces allocations requiring small numbers of dubious chores to become envy-free.
    This beats the theoretical existence guarantee of \DEF{(n-1)} by half. \DEF{k} also appears to correlate well with \PO{}, which optimally require few dubious chores to become envy-free.
\end{itemize}


\section{Related Work}
\label{sec:related_work}

While the literature on fair division of indivisible items is vast (see e.g., \citet{ABF+22fair} for a survey), our work fits best in the contexts of algorithms guaranteeing fair and efficient allocations (often for restricted preference classes), fairness achieved with copies, and fairness achieved by imposing epistemic constraints. 
The following related works hold for agents with additive valuations; see Section \ref{sec:prelims} for formal definitions.

\subsection{Algorithms Yielding Fairness and Efficiency}
For goods instances and agents with general valuations, \citet{murhekar2021fair} identified a pseudo-polynomial time algorithm for computing \EF{1} and \fPO{} 
that extends to polynomial time by fixing the numbers of agents or different values that agents have over goods. This improves algorithms by \citet{aziz2022fair} and \citet{barman2018finding} which are polynomial time for two agents with general valuations and if valuations are bounded, respectively. \citet{garg2021computing} identified a polynomial time algorithm for computing EFX, a strict generalization of \EF{1}, and \fPO{}, for bivalued goods. The polynomial time algorithm of \citet{gorantla2023fair} computes EFX for two-types of goods. 
\citet{garg2023improving}'s polynomial time algorithm identifies \EF{1}+\fPO{} for \emph{two-type} agents, where each agent has one of two utility functions. 
Still, computing \EF{1}+\fPO{} in polynomial time for the general goods case remains an open question.

For chores instances, 
\citet{garg2022fair} and \citet{ebadian22:fairly} independently identified polynomial time algorithms to compute \EF{1}$+$\fPO{} allocations when agents' valuations only take on one of two possible values.
Similarly, \citet{aziz2023twotypes} provided an algorithm to compute \EF{1}$+$\fPO{} allocations when there are only two types of chores.
\citet{garg2023improving} further identified polynomial time algorithms that compute an allocation that is \EF{1} and \fPO{} for either three agents, two-type agents, or \emph{personalized} bivalued chores instances---bivalued chores in which agents may have different values---with different enough valuations.

\subsection{Epistemic Fairness and Copies}
%
Our work aligns closely with 
\textit{epistemic} fair division which describes
agents 
with
limited information
about allocations.
\citet{ABC+18knowledge}, and \citet{caragiannis2023existence} define epistemic \EF{} and EFX 
as
allocations in which each agent, based on information about their own bundle only, believes it is possible 
to allocate the remaining items to other agents
such 
that the overall allocation satisfies \EF{} or EFX.
Thus, each agent's view may be substantially different from each other or the real underlying allocation. 
The notion of hidden envy-freeness of \citet{Hosseini2020:Fair} supposes that each agent agrees on the visible information, but their belief about what's hidden may differ.
Other 
related works include a Bayesian approach to incomplete information \citep{chen2017ignorance}, a generalization of maximin share \citep{chan2019maximin}, and where agents only compare their values to their neighbors on a social network \citep{abebe2017fair,bei2017networked,beynier2019local,bredereck2022envy,chevaleyre2017distributed}. \citet{bliznets2023fair}
combine these approaches to study \EF{} for agents embedded in a social network who can hide at most one good in their allocation, and at most $k$ goods are hidden.

The concept of ``copies'' of items appears in several works,
but in very different settings to ours.
For example, \citet{akrami2023fair} considered actual copies of chores that are added to allocations
to obtain fair solutions. 
In this case, each actual copy of a chore creates
additional disutility for the agents.
Therefore, their approach for chores resembles the
literature on goods allocation with charity,
in which a small number of goods may be left unallocated
in order to obtain a fair solution. Literature on EFX with charity was introduced by \citet{caragiannis2019envy} and followed-up with many relaxations \citep{chaudhury2021little,berger2022almost,akrami2022ef2x,berendsohn2022fixed,mahara2023extension}.
On a technical level, \citet{akrami2023fair} develops an algorithm that yields an \EF{1} and \fPO{} allocation with $n-1$ surplus, whereas our algorithms guarantees \EF{} and \fPO{} upon adding $2n-2$ dubious chores to the original instances.

Separately, \citet{aleksandrov2022envy} considered
modified definitions of \EF{1} and EFX in which a counterfactual
``copy'' of an item is added to a bundle
to eliminate pairwise envy. 
\citet{gorantla2023fair} studied the setting of goods divided into several types such that each agent values goods of the same type identically.
Hence, goods of each type can be seen as ``copies,''
but they are given exogenously, in contrast to our setting.
Likewise, \citet{aziz2023twotypes}
focused on the setting
with two types of chores.

\section{Preliminaries}
\label{sec:prelims}

\paragraph{Problem instance.}
An \emph{instance} $\calI = \langle \calN, \calM, \calV \rangle$ of the fair division problem is defined by a set of $n$ \emph{agents} $\calN$, a set of $m$ \emph{chores} $\calM$, and a \emph{valuation profile} $\calV = (v_i)_{i \in \calN}$ that specifies the preferences of every agent $i \in \calN$ over each subset of the chores in $\calM$ via a \emph{valuation function} $v_i: 2^{\calM} \rightarrow \mathbbm{R}$. If $v_i(S) \geq 0 ~\forall i \in \calN, S \subseteq \calM$, then we call the items \emph{goods}; likewise, if $v_i(S) \leq 0$ the items are \emph{chores}.
We assume that the valuation functions are \emph{additive}, i.e., for any $i \in \calN$ and $S \subseteq \calM$, $v_i(S) = \sum_{c \in S} v_i(\{c\})$, where $v_i(\emptyset) = 0$. 
We write $v_i(c)$ instead of $v_i(\{c\})$ for a single chore $c \in \calM$. 

\paragraph{Allocation.}
A real \emph{allocation} $A = (A_1,\dots,A_n)$ refers to an $n$-partition of the set of chores $\calM$, where $A_i \subseteq \calM$ is the \emph{bundle} allocated to agent $i \in \calN$ and no chore of $\calM$ is unallocated. The utility of agent $i$ for the bundle $A_i$ is given by $v_i(A_i) = \sum_{c \in A_i} v_{i}(c)$.
An allocation is \emph{fractional} if chores may be shared between agents 
and \emph{integral} otherwise; each chore is nevertheless fully allocated across the agents.

\paragraph{Restricted valuations.}

We consider four special cases of agent valuation functions.
The instance $\calI$
has \emph{identical valuation} if for any two agents $i, j \in \calN$
and chore $c \in \calM$ it holds that $v_{i}(c) = v_{j}(c)$.
It has \emph{binary valuations} if $v_{i}(c) \in \{-1,0\}$
for every agent $i \in \calN$ and chore $c \in \calM$.
The instance has \emph{bivalued valuations}
if there exist two real numbers $x < y < 0$ such that
$v_{i}(c) \in \{x,y\}$
for every $i \in \calN$ and $c \in \calM$.
Finally, the instance has \emph{two types of chores}
if $\calM$ can be partitioned into two sets $X$ and $Y$
such that, for every agent $i \in \calN$
and two chores $c, c'$ from the same set,
it holds that $v_i(c) = v_i(c')$.

\begin{definition}[\textbf{Dubious chores}]
A \emph{dubious} chore $c'$ allocated to agent $i \in \calN$ is a \emph{copy} of the chore $c \in \calM$ that upholds the same perceived costs for all agents but no incurred cost for $i$. That is, $v_{h}(c') = v_{h}(c)~ \forall h \in \calN\setminus\{i\}$ and $v_{i}(c') = 0$.
Given an instance $\calI$ and real allocation $A$, a \emph{dubious multiset} $D$ refers to a multiset containing dubious chores copied from $\calM$. A \emph{dubious allocation} $A^D =(A_1^{D}, \ldots, A_n^{D})$ is an $n$-partition of the multiset $D$. 

We define the \emph{augmented allocation} $A^* = A \cup A^D$ such that for each $i \in \calN$, $A^*_i = A_i \cup A^D_i$.
The utility of agent $i$ for its own augmented bundle is $v_i(A^*_i) = v_i(A_i)$ while for other agents $h\in \calN$ is $v_h(A^*_i) = v_h(A_i) + v_h(A^D_i)$.
\label{definition:dubious_chore}
\end{definition}

\begin{definition}[\textbf{Envy-freeness}]
	An allocation $A$ (real, dubious, or augmented) is \emph{envy-free} (\EF{}) if for every pair of agents $i,h \in \calN$, $v_i(A_i) \geq v_i(A_h)$. 
	The allocation $A$ is \emph{envy-free up to one chore} (\EF{1}) if for every pair of agents $i,h \in \calN$ such that $A_i \neq \emptyset$, there exists some chore $c \in A_i$ such that $v_i(A_i \setminus \{c\}) \geq v_i(A_h)$.
        Furthermore, $A$ is \emph{envy-free up to any chore} (\EFX{}) if for every pair of agents $i,h \in \calN$ and chore $c \in A_i$ such that $v_i(c)<0$, it holds that $v_i(A_i \setminus \{c\}) \geq v_i(A_h)$.
\end{definition}

\begin{definition}[\textbf{Envy-freeness with dubious chores}]
A real allocation $A$ is said to be \emph{envy-free up to $k$ dubious chores} (\DEF{k}) if there exists a dubious multiset $D$ and dubious allocation $A^D$ such that: (i) $D$ consists of up to $k$ dubious chores copied from $\calM$, and (ii) the augmented allocation $A^* = A \cup A^D$ is envy-free. 
\end{definition}

\begin{remark}
    It follows from the above definitions that an allocation is \EF{} if and only if it is \DEF{0}. 
    Moreover, for any $k \geq 0$, a \DEF{k} allocation is also \DEF{(k+1)}, but the converse may not hold.
\end{remark}

In the full version of the paper \citep{hosseini2023distribution},
we also consider two stronger notions of \DEF{k}: \emph{singly} and \emph{personalized}.
There, we introduce restrictions on how dubious chores may be allocated that may be well motivated in some of the potential applications.

\paragraph{Pareto optimality.}
An allocation $A$ is \emph{Pareto dominated} by allocation $B$ if $v_i(B_i) \geq v_i(A_i)$ 
$\forall i \in \calN$, with at least one of the inequalities being strict. A \emph{Pareto optimal} (\PO{}) allocation is one that is not Pareto dominated by any other allocation.
A \emph{fractional Pareto optimal} (\fPO{}) allocation is one that is not Pareto dominated by any other fractional allocation.

Note that an \fPO{} allocation is also \PO, but the converse may not necessarily hold. \PO{} does imply \fPO{} for special cases such as bivalued preferences \citep{ebadian22:fairly}.

\begin{example}
Consider the instance $\calI$ in Table \ref{tab:prelims_ex} with $n$ agents and $m=n$ chores defined with agents' valuations $v_{i}(c)$ for $i \in \calN, c \in \calM$. 
A real allocation $A$ is identified by the circled valuations such that $A_i = \{c_i\},~\forall i \in \calN$. Clearly $A$ is \EF{1} and \PO{} and agent $n$ envies each other agent $i \neq n$. 

Consider the dubious multiset $D$ containing $n-1$ dubious copies of $c_n$ that is allocated as $A^D = (\{c_n'\}, \{c_n'\}, \ldots, \{c_n'\}, \emptyset)$. Then, the augmented allocation $A^* = A \cup A^D$ is envy-free. This follows because for any pair of agents $i,h \in \calN \backslash \{n\}$, 
we have that $v_n(A_i^*) = -n-1 < -n = v_n(A_n^*)$, $v_i(A_h^*) = -n-1 < -1 = v_i(A_i^*)$, and $v_i(A_n^*) = -n < -1 = v_i(A_i^*)$. 
Hence, $A$ is \DEF{(n-1)}. 

\label{ex:prelims_ex}
\end{example}

\begin{table}[t] 
    \setlength{\tabcolsep}{4pt}
    \centering
    \begin{tabular}{c|c c c c c | c c c c}
        agent &  $c_1$ & $c_2$ & $\ldots$ & $c_{n-1}$ & $c_n$ &
        $c_n'$ & $c_n'$ & $\ldots$ & $c_n'$\\
        \hline
        $1$ 
            & $\circled{-1}$ & $-1$ & $\ldots$ & $-1$ & $-n$ 
            & $\dashedcirc{0}$ & $-n$ & $\ldots$ & $-n$\\
        $2$
            & $-1$ & $\circled{-1}$ & $\ldots$ & $-1$ & $-n$
            & $-n$ & $\dashedcirc{0}$  & $\ldots$ & $-n$\\
        $\vdots$ & & $\vdots$ & $\ddots$ & & $\vdots$ 
            & $\vdots$ & & $\ddots$ & $\vdots$\\
        $n-1$ 
            & $-1$ & $-1$ & $\ldots$ & $\circled{-1}$ & $-n$ 
            & $-n$ & $-n$  & $\ldots$ & $\dashedcirc{0}$\\
        $n$ 
            & $-1$ & $-1$ & $\ldots$ & $-1$ & \circled{-$n$}
            & $-n$ & $-n$  & $\ldots$ & $-n$\\
    \end{tabular}
    \caption{An example allocation and its dubious augmentation.}
    \label{tab:prelims_ex}
\end{table}

\begin{remark}
Every real allocation $A$ is \DEF{(m \cdot (n-1))} since it can be augmented with a dubious multiset $D$ with $n-1$ dubious copies of each chore. Allocating these one per agent, except for the agent that receives the corresponding real chore in $A$, will make $A^* = A \cup A^D$ envy-free. However, this is a trivial allotment akin to only revealing each agent's own bundle to themselves, as in epistemic envy-freeness \citep{ABC+18knowledge,caragiannis2023existence}. In this work we provide computational limits and theoretical bounds on the minimum $k$ for which \DEF{k} allocations exist.
\end{remark}

\begin{table*}[t!]
\setlength{\tabcolsep}{1.75pt}
\centering
\footnotesize
\begin{tabular}{cc|cccc cccc cc cccc cccc cc cc cc cc cc cc | cc cc cc cc} 
agent &
& $d^1_{1,1}$ & $\tightdots$ & $d^4_{1,1}$ &
& $d^1_{1,2}$ & $\tightdots$ & $d^4_{1,2}$ &
& $\cdots$ &
& $d^1_{1,n}$ & $\tightdots$ & $d^4_{1,n}$ &
& $d^1_{2,1}$ & $\tightdots$ & $d^4_{2,1}$ &
& $d^1_{2,2}$ &
& $\cdots$ &
& $d^4_{2,n}$ &
& $d^1_{3,1}$ &
& $\cdots$ &
& $d^4_{n,n}$ &
& $c_{1}$ && $c_{2}$ && $\cdots$ && $c_{m}$ \\
\hline
$0$ &
& $0$ & $\tightdots$ & $0$ &
& $0$ & $\tightdots$ & $0$ &
& $\cdots$ &
& $0$ & $\tightdots$ & $0$ &
& $0$ & $\tightdots$ & $0$ &
& $0$ &
& $\cdots$ &
& $0$ &
& $0$ &
& $\cdots$ &
& $0$ &
& $\circled{0}$ && $\circled{0}$ && $\cdots$ && $\circled{0}$
\\
$1$ &
& $\circled{-1}$ & $\tightdots$ & $\circled{-1}$ &
& $-1$ & $\tightdots$ & $-1$ &
& $\cdots$ &
& $-1$ & $\tightdots$ & $-1$ &
& $\circled{0}$ & $\tightdots$ & $\circled{0}$ &
& $0$ &
& $\cdots$ &
& $0$ &
& $\circled{0}$ &
& $\cdots$ &
& $0$ &
&  && &&  && 
\\
$2$ &
& $0$ & $\tightdots$ & $0$ &
& $\circled{0}$ & $\tightdots$ & $\circled{0}$ &
& $\cdots$ &
& $0$ & $\tightdots$ & $0$ &
& $-1$ & $\tightdots$ & $-1$ &
& $\circled{-1}$ &
& $\cdots$ &
& $-1$ &
& $0$ &
& $\cdots$ &
& $0$ &
&  &&  &&  && 
\\
$3$ &
& $0$ & $\tightdots$ & $0$ &
& $0$ & $\tightdots$ & $0$ &
& $\cdots$ &
& $0$ & $\tightdots$ & $0$ &
& $0$ & $\tightdots$ & $0$ &
& $0$ &
& $\cdots$ &
& $0$ &
& $-1$ &
& $\cdots$ &
& $0$ &
& $S_1$ && $S_2$ && $\cdots$ && $S_m$
\\
$\vdots$ &
&  & $\vdots$ &  &
&  & $\vdots$ &  &
&  &
&  & $\vdots$ &  &
&  & $\vdots$ &  &
&  &
& $\vdots$ &
&  &
&  &
& $\vdots$  &
&  &
&  &&  &&  && 
\\
$n$ &
& $0$ & $\tightdots$ & $0$ &
& $0$ & $\tightdots$ & $0$ &
& $\cdots$ &
& $\circled{0}$ & $\tightdots$ & $\circled{0}$ &
& $0$ & $\tightdots$ & $0$ &
& $0$ &
& $\cdots$ &
& $\circled{0}$ &
& $0$ &
& $\cdots$ &
& $\circled{-1}$ &
&  &&  &&  && 
\end{tabular}
\caption{An illustration to the proof of \cref{thm:veryfying:defk:hardness}. }
\label{table:thm:veryfying:defk:hardness}
\end{table*}

To conclude this section, we recall three techniques for allocating chores to agents that we use in our subsequent theorems.

\paragraph[Round Robin algorithm.]{Round Robin algorithm.} 
Fix a permutation $\sigma$ of the agents. The \texttt{Round Robin} algorithm cycles through the agents according to $\sigma$. In each round, an agent gets its favorite (i.e., least undesirable) chore from the pool of remaining chores.

\paragraph[Envy Graph algorithm.]{Envy Graph algorithm \citep{Lipton04:Approximately}.}
The \texttt{Envy Graph} algorithm iterates through each chore $c$ in rounds by first resolving cycles in the \emph{top trading envy graph} $T_A$ based on the partial allocation $A$. This graph is built over $n$ vertices and composed of edges $(i,h)$ if agent $i$'s (weakly) most preferred bundle is $h$ in $A$. The algorithm resolves cycles in $T_A$, by transferring to each agent $i$ the bundle of the next agent in the cycle. Afterwards, one agent who does not envy any other agent is given $c$.

\paragraph{Fisher Markets.}
A \emph{pricing vector} is a function $p : \calM \rightarrow \mathbbm{R}_{\ge 0}$.
Intuitively, the price $p(c)$ of chore $c \in \calM$ is a payment that an agent receives for doing the chore.
For a subset $S \subseteq \calM$, we have $p(S) = \sum_{c \in S} p(c)$.
Given $p$, agent $i$'s \emph{minimum pain per buck} 
is the set of chores with the minimum absolute value of 
valuation to price ratio:
\(
    MPB_i = \arg \min_{c \in \calM} |v_i(c)|/p(c);
\)
intuitively, the most profitable chores for $i$.

A pair of a real allocation and a pricing vector $(A,p)$
is a \emph{Fisher market equilibrium} if every agent $i \in \calN$
receives only its minimal pain per buck chores, i.e.,
$A_i \subseteq MPB_i$.
It is well-known that every allocation of a Fisher market equilibrium is \fPO{} \citep{Mas-Colell95:Microeconomic}. 
A Fisher market equilibrium $(A,p)$ is \emph{price envy-free up to one item} (\pEF{1})
if for every $i, j \in \calN$, either $A_i = \emptyset$
or there exists $c \in A_i$ such that $p(A_i \setminus \{c\}) \le p(A_j)$.
It is known that if $(A,p)$ is \pEF{1},
then $A$ is \EF{1} \citep{ebadian22:fairly,garg2022fair}.

\section{Fairness with Dubious Chores}
\label{sec:fairness}

We begin by analyzing fairness with dubious chores without any efficiency requirement. We show the computational hardness of finding the smallest $k$ for which \DEF{k} exists, discuss the relation between \DEF{k} and \EF{1}, and prove that an existing algorithm can achieve \DEF{(n-1)} in polynomial time.
Later, in \cref{sec:fairness_and_efficiency}, we study the existence and computation of \DEF{k} along with \fPO{}.

We start with the decision problem of finding the optimal $k$ for which a \DEF{k} allocation exists. \citet{BSV21approximate} demonstrated that determining whether for a given instance there exists an \EF{} allocation is NP-complete.
We show that for an arbitrary \textit{fixed} constant $k \in \mathbbm{Z}_{\ge 0}$,
determining whether an instance has a \DEF{k} allocation is NP-complete, even in the cases of identical or binary valuations.

\begin{restatable}{theorem}{thmexstdefknp}
\label{thm:exst:defk:np}
Given an instance $\calI = \langle \calN, \calM, \calV \rangle$,
for every fixed constant $k \in \mathbbm{Z}_{\ge 0}$,
deciding if the instance admits a \DEF{k} allocation $A$ is NP-complete,
even when valuations are identical or binary.
\end{restatable}

The proof uses the techniques
developed by \citet{BSV21approximate} and \citet{Hosseini2020:Fair}
and can be found in the full version of the paper \citep{hosseini2023distribution}.
Observe that \cref{thm:exst:defk:np} implies
that finding the minimal $k \in \mathbbm{Z}_{\ge 0}$
such that there exists a \DEF{k} allocation is NP-hard as well.
Moreover, since hardness holds also for $k=0$,
there is no polynomial-time constant approximation scheme for this problem,
unless P = NP.

\begin{coro}
\label{cor:no_approx:defk}
Given instance $\calI = \langle \calN, \calM, \calV \rangle$, unless P=NP, there is no polynomial-time algorithm that gives a constant approximation for the problem of finding \DEF{k} allocation with minimal $k$.
\end{coro}

\cref{thm:exst:defk:np} holds when $k$ is a fixed constant.
However, if we do not fix $k$,
even the verification problem becomes hard.
This means that the natural question of
whether a given allocation
might be augmented with $k$ dubious chores
to an envy-free allocation
is computationally intractable.

\begin{theorem}
\label{thm:veryfying:defk:hardness}
Given an instance $\calI = \langle \calN, \calM, \calV \rangle$
with binary valuations and an allocation $A$,
deciding if allocation $A$ is \DEF{k}
is NP-complete.
\end{theorem}
\begin{proof}
Given a dubious allocation $A^D$,
we can verify if the augmented allocation $A^* = A \cup A^D$
is \EF{} in polynomial time, so our problem is in NP.
We therefore focus on showing the hardness.

To this end, we follow the reduction from \textsc{Restricted Exact Cover by 3-Sets} (RX3C).
In RX3C instance $\calI = \langle \mathcal{U}, \mathcal{S} \rangle$,
we are given a universe of $n = 3k$ elements $\mathcal{U} = (u_1,u_2,\dots,u_n)$
and a family $\mathcal{S} = (S_1,S_2,\dots,S_m)$ of three-element subsets of $\mathcal{U}$
such that every element $u_i$ appears in exactly 3 subsets, i.e.
$|S_i \in \mathcal{S} : u_j \in S_i| = 3$ for every $j \in [n]$.
Hence, necessarily $m = 3k = n$.
The question is whether we can find a cover of size $k$, i.e.,
a subfamily of subsets $\mathcal{K} \subset \mathcal{S}$
such that $|\mathcal{K}| = k$ and
$\bigcup_{S_i \in \mathcal{K}} S_i = \mathcal{U}$.
Observe that in such a case
every element of $\mathcal{U}$ will appear in
subsets in $\mathcal{K}$ exactly once,
with means that $\mathcal{K}$ will be
an exact cover as well.
The problem is known to be NP-hard \citep{gonzalez1985clustering}.

Now, for every instance $\calI$ of RX3C,
we construct a corresponding instance
$\calI' = \langle \calN, \calM, \calV \rangle$
and an allocation $A$
(for an illustration see \cref{table:thm:veryfying:defk:hardness}).
First, let us take one agent for each element of $\mathcal{U}$
and one \emph{choosing} agent $0$.
Formally, $\calN = \{0,1,\dots, n\}$.
Next, for every ordered pair of agents $(i,j) \in \{1,\dots,n\}^2$
let us take $4$ \emph{dummy chores},
$d^1_{i,j}, d^2_{i,j}, d^3_{i,j},$ and $d^4_{i,j}$.
Also, let us take one chore for every subset in $\mathcal{S}$,
i.e., chores $c_1,\dots,c_m$.
This will give us a total of $4n^2 + m$ chores in $\calM$.
Further, let us describe the valuations of the agents.
For the choosing agent $0$,
we set the value of every chore to $0$.
In turn, for every agent $i \in \{1,\dots,n\}$,
we set its valuation of dummy chores to $-1$,
if $i$ is the first agent in the subscript,
and $0$, otherwise.
Formally,
for every $r \in \{1,2,3,4\}$ and  $(j, k) \in \{1,\dots,n\}^2$
let $v_i(d^r_{j,k}) = -1$, if $j=i$,
and $v_i(d^r_{j,k}) = 0$, otherwise.
As for chores $c_1,\dots,c_m$
the valuation depends on whether element $u_i$ belongs to
the subset corresponding to the particular chore.
Formally,
we set $v_i(c_j) = -1$, if $u_i \in S_j$, and $v_i(c_j) = 0$, otherwise.
Finally, let us specify the allocation $A$.
First, we give all of the chores $c_1,\dots,c_m$ to the choosing agent $0$,
i.e., $A_0 = \{c_1,\dots,c_m\}$.
Then, every agent $i \in \{1,\dots,n\}$,
receives all of the dummy chores in which it is in the second position of the subscript,
i.e., $A_i = \bigcup_{r =1}^4 \{d^r_{1,i},d^r_{2,i},\dots,d^r_{n,i}\}$.

In the remainder of the proof, let us show that 
$A$ is \DEF{k}, if and only if,
there exists an exact cover $\mathcal{K}$ in the original instance $\calI$.
To this end, observe that for every two agents $i,j \in \{1,2,\dots,n\}$
we have $v_i(A_j) = v_i(A_i) = -4$.
Hence, there is no envy among these agents.
Furthermore, for the choosing agent $0$, we have $v_0(S) = 0$ for every $S \subseteq \calM$,
thus it does not envy any other agents.
However, for every agent $i \in \{1,2,\dots,n\}$
we have that $v_i(A_0)$ is equal to the number of subsets in $\mathcal{S}$
that include element $u_i$ times $-1$, which is $-3$.
Hence, the choosing agent is envied by every other agent.
Now, in order to eliminate such envy,
we have to give the choosing agent $0$
at least one dubious 
copy of a chore that has value $-1$
for each of agents $1,\dots,n$.
Since we have $3k$ such agents,
each chore has value $-1$ for maximally $3$ agents, and
we can add up to $k$ dubious chores,
we have to choose a subset
of $k$ chores from $\{c_1,c_2,\dots,c_m\}$
in such a way that every agent from $\{1,\dots,n\}$
has value $-1$ for at least one of them.
But that is possible if and only if
there exists a set cover in the original instance $\calI$.
\end{proof}

In lieu of this computational hardness,
we establish upper-bounds on the required number of dubious chores to make a real allocation envy-free.
While \EF{1} allocations of chores do exist and can be computed in polynomial time \citep{aziz2022fair,BSV21approximate}, such allocations may require many (i.e., $n(n-1)$) dubious chores to become envy-free. 
This is because for each pairwise envy relation between agents, the envied agent can dubiously receive the ``worst'' chore of the envious agent
to remove envy.
Example \ref{ex:ef1:defnn-1} demonstrates that this bound is in fact tight.

\begin{proposition}
\label{prop:ef1:defnn-1}
Given an instance $\calI = \langle \calN, \calM, \calV \rangle$,
every \EF{1} allocation $A$ is
also \DEF{n(n-1)}. 
\end{proposition}

\begin{example}
\label{ex:ef1:defnn-1}
The instance presented in Table \ref{table:example:def-n(n-1)} demonstrates an 
\EF{1} allocation $A$ that is \DEF{n(n-1)} but not \DEF{k} for any $k < n(n-1)$. Note that $A$ is extremely ineffective---any other allocation would yield a Pareto improvement. Furthermore, we note that this allocation could be an output of the \texttt{Envy Graph} algorithm.

\begin{table}[t!] 
    \centering
    \setlength{\tabcolsep}{2pt}
    \begin{tabular}{c|c c c c} 
        agent & $c_1$ & $c_2$ &  $\cdots$ & $c_n$  \\
        \hline
        $1$ & $\circled{-1}$ & $0$ & & $0$ \\ 
        $2$ & $0$ & $\circled{-1}$ & & $\cdots$ \\ 
        $\vdots$ &  & $\vdots$ &  $\ddots$ & $\vdots$ \\ 
        $n$ & $0$ & $0$ & $\cdots$ & $\circled{-1}$
    \end{tabular}
    \caption{\EF{1} allocation that is not \DEF{k} for any $k < n(n-1)$.}
    \label{table:example:def-n(n-1)}
\end{table}
\end{example}

We next demonstrate stronger results for allocations produced by \texttt{Round Robin}, which require at most $n-1$ dubious chores to become envy-free. The improvement arises since each agent chooses the best of the remaining chores for each round.
Hence, in each round, there is a single agreed-upon ``worst'' chore, the latest chore allocated, that can eliminate all envy-relations when dubiously allocated to all envious agents.

\begin{theorem}
\label{thm:round-robin:defn-1}
Given an instance $\calI = \langle \calN, \calM, \calV \rangle$,
\texttt{Round Robin} returns a \DEF{(n-1)} allocation. 
\end{theorem}
\begin{proof}
Without loss of generality, let $A$ be a real allocation generated by \texttt{Round Robin} by the order $\sigma = \{1,2,\ldots,n\}$. Let $T$ be the last agent allocated a chore and let $T' = (T+1) ~mod~ n$ be the subsequent agent in $\sigma$. We know that $A$ is \EF{1} and, specifically $\forall i \neq T'$ and $k \neq i$, $v_i(A_i \backslash\{\tilde{c}_i\}) \geq v_i(A_k)$, where $\tilde{c}_i$ is the last chore allocated to each agent $i$ by \texttt{Round Robin} \citep{aziz2022fair,Caragiannis19:unreasonable}.
The subsequent agent $T'$ is not envious of any other agent $i$,
since for every chore $c$ of $T'$
there is one chore of $i$ that was chosen after $c$.
We also know that in the last $n$ rounds of \texttt{Round Robin}, each agent $i \neq T$ chose their chore $\tilde{c}_i$ over the last allocated chore $\tilde{c}_T$, meaning that $v_i(A_i) \geq v_i(A_i \cup \{\tilde{c}_T\} \backslash \{\tilde{c}_i\})$. Therefore for every agent $i \neq T'$ and $k \neq i$, we have $v_i(A_i) \geq v_i(A_k \cup \{\tilde{c}_i\}) \geq v_i(A_k \cup \{\tilde{c}_T\})$. As a result, if a dubious copy of $\tilde{c}_T$ were given to every agent except $T$, no agent would envy any other. This is $n-1$ dubious chores.
\end{proof}

Allotting $n-1$ dubious chores for allocations produced by \texttt{Round Robin} is similar to the case of goods allocations satisfying \emph{strong envy-freeness up to one good} (\sEF{1}) \citep{conitzer2019group}, in which all envy to each agent $i$ can be eliminated by removing a single good in $i$'s bundle $A_i$. \citet{Hosseini2020:Fair} observe that hiding $n-1$ of these goods can make \sEF{1} allocations envy-free and that this bound is tight for the existence of any hidden envy-free allocation.
Likewise, with chores, \cref{ex:prelims_ex} above demonstrates that the existence of \emph{any} \DEF{k} allocation is tight at $k=n-1$. This is because some agent $i$ must receive $c_n$ in the real allocation and they will envy every other agent by a value of at least $1$. Therefore, at least $n-1$ chores must be dubiously copied to make $i$ not envious, so any real allocation cannot be \DEF{(n-2)}.

\citet{conitzer2019group} discuss that both \texttt{Round Robin} and \texttt{Envy Graph} yield \sEF{1} allocations for goods; however, only \texttt{Round Robin} of these yields \DEF{(n-1)} allocations for chores, which are notably also \EF{1}.
Example \ref{ex:ef1:defnn-1} previously demonstrated that \EF{1} allocations furnished by 
\texttt{Envy Graph} may be not \DEF{(n-1)}.
Finally, 
\DEF{(n-1)} allocations are 
not necessarily
\EF{1}, as in Example \ref{ex:DEF-not-EF1}. 

\begin{example}
Fix $n=3$ and $m=4$ in the instance depicted by Table \ref{tab:DEF-not-EF}. The real allocation indicated by circled valuations is \DEF{2} with respect to the dubious set $D=\{c'_1,c'_2\}$ and allocation $A^D = \{\emptyset, \{c_1', c_2'\}, \emptyset\}$, but it is not \EF{1}.
\label{ex:DEF-not-EF1}
\end{example}

\section{Fairness and Efficiency}
\label{sec:fairness_and_efficiency}

\begin{table}[t!] 
    \centering
    \renewcommand{\arraystretch}{1.5}
    \setlength{\tabcolsep}{2pt}
    \begin{tabular}{c|c c c c | c c}
         agent &  $c_1$ & $c_2$ & $c_3$ & $c_4$ & $c_1'$ & $c_2'$ \\
        \hline
        $1$ 
            & $\circled{-3}$ & $\circled{-2}$ & $-1$ & $-5$ & $-3$ & $-2$ \\
        $2$
            & $-1$ & $-1$ & $\circled{-1}$ & $-1$ & $\dashedcirc{0}$ & $\dashedcirc{0}$\\
        $3$ 
            & $-1$ & $-1$ & $-1$ & $\circled{-1}$ & $-1$ & $-1$
    \end{tabular}
    \caption{An example allocation that is \DEF{2} but not \EF{1}.}
    \label{tab:DEF-not-EF}
\end{table}

In this section, we proceed to demonstrate the existence and computation of \DEF{k} along with fractional Pareto optimality.
Specifically, we show that 
we can efficiently compute \DEF{(n-1)} and \fPO{} allocations
in four special cases: if agents have identical, binary, or bivalued valuations,
or when there are two types of chores.
However, we start by a general result that
a \DEF{(2n-2)} and \fPO{} allocation always exist.
Our approach is based on the algorithm for finding
fair and efficient allocations for goods developed by \citet{barman2019proximity}.
We note the same algorithm was utilized by \citet{akrami2023fair}
to attain an \EF{1} and \fPO{} allocation by introducing at most $n-1$ actual copies of chores.


\begin{theorem}
\label{thm:def2n-2+po}
Given an instance $\calI = \langle \calN, \calM, \calV \rangle$,
there always exists an \DEF{(2n-2)} and \fPO{} allocation.
\end{theorem}
\begin{proof}
\citet{barman2019proximity} showed%
\footnote{
\citet{barman2019proximity} only considered goods
but the same algorithm works for chores~\citep{BraSan-2023-CEEIchores}.
However, the algorithm is based on finding
\emph{competitive equilibrium with equal incomes} (CEEI),
which can be computed in polynomial time for goods.
Only an XP algorithm for chores parameterized by the number of agents or chores is known; the existence of nonparameterized polynomial-time algorithm is
unknown.}
that there always exists
a Fisher market equilibrium $(A,p)$ such that
for every agent $i \in \calN$,
\begin{align}
    \label{eq:def2n-2+po}
    \notag
    p(A_i) &\le 1 \mbox{ and }
        \exists c \in \calM \mbox{ such that }
        p(A_i \cup \{c\}) \ge 1, \mbox{ or}\\
    p(A_i) &> 1 \mbox{ and }
        \exists c \in A_i \mbox{ such that }
        p(A_i \setminus \{c\}) < 1.
\end{align}
Since $(A,p)$ is an equilibrium,
we know that $A$ is \fPO{}.
In what follows, we show that $A$ is
\DEF{(2n-2)} as well.

To this end, let us denote by $i_{\max} \in \arg \max_{i \in \calN} p(A_i)$ an agent with maximally priced bundle
and by $c_{\max} \in \arg \max_{c \in \calM} p(c)$
a chore with the highest price.
Let $D$ consist of $2n-2$ dubious copies of chore $c_{\max}$,
which we denote $c_{\max}'$,
and let $A^D$ be a
dubious allocation in which
each agent in $\calN \setminus \{i_{\max}\}$
receives two dubious copies of this chore.
Next we show the augmented allocation
$A^* = A \cup A^D$ is envy-free. 

Observe that every agent $i \in \calN$
does not envy agent $i_{\max}$, even in the real allocation $A$.
Indeed, since $i$ receives only its minimum pain per buck chores,
\begin{align*}
    v_i(A_i) &= - p(A_i) \cdot \min_{c \in \calM} \textstyle\frac{|v_i(c)|}{p(c)}
        \ge p(A_{i_{\max}}) \cdot \max_{c \in \calM} \textstyle\frac{v_i(c)}{p(c)}\\
        &\ge \textstyle\sum_{c \in A_{i_{\max}}} p(c) \cdot \textstyle\frac{v_i(c)}{p(c)} 
        = v_i(A_{i_{\max}}).
\end{align*}

Hence, in $A^*$ no agent envies $i_{\max}$ as well.
Now, let us take arbitrary $i,j \in \calN$ such that $i \neq i_{\max}$
and show that $j$ does not envy $i$.
Observe that
\begin{align*}
    v_j(A^*_i) &= v_j(A_i) + 2 v_j(c_{\max}') \\
    &=\textstyle\sum_{c \in A_i} \left( p(c) \textstyle\frac{v_j(c)}{p(c)} \right)
        + 2 p(c_{\max}) \textstyle\frac{v_j(c_{\max}')}{p(c_{\max})}\\
    &\le - (p(A_i) + 2 p(c_{\max})) \cdot \min_{c \in \calM} \textstyle\frac{|v_j(c)|}{p(c)}.
\end{align*}

From \cref{eq:def2n-2+po} we know that $p(A_i) + 2 p(c_{\max}) \ge p(A_j)$.
Hence,
$
    v_j(A^*_i) \le - p(A_j) \cdot  \min_{c \in \calM} \frac{|v_j(c)|}{p(c)}
                = v_j(A_j)
                = v_j(A^*_j).
$
Thus, $A^*$ is envy-free, so $A$ is \DEF{(2n-2)}.
\end{proof}


Next, we focus on four restricted domains of preferences.
We utilize their structure to strengthen our result and prove
the existence of algorithms that find
\DEF{(n-1)} and \fPO{} allocations in polynomial time.
We begin by considering identical and then binary valuations.

\begin{theorem}
\label{thm:defn-1+po:identical}
For every instance $\calI = \langle \calN, \calM, \calV \rangle$
with identical valuations,
there exists a \DEF{(n-1)} and \fPO{} allocation
and it can be computed in polynomial time.
\end{theorem}
\begin{proof}
Observe that with identical valuations every allocation is \fPO{}.
Hence, by \cref{thm:round-robin:defn-1},
\texttt{Round Robin} algorithm returns a
\DEF{(n-1)} and \fPO{} allocation in such case.
\end{proof}


\begin{theorem}
\label{thm:defn-1+po:binary}
For every instance $\calI = \langle \calN, \calM, \calV \rangle$
with binary valuations,
there exists a \DEF{(n-1)} and \fPO{} allocation
and it can be computed in polynomial time.
\end{theorem}
\begin{proof}
The following algorithm returns a \DEF{(n-1)} and \fPO{} allocation.
First, while possible, assign every chore to any agent that values it at $0$.
This ensures that the final allocation will
have the minimal possible total cost for the agents,
also among fractional allocations.
Thus, it will be \fPO{}.
Let us assign 
the remaining \emph{common chores}, i.e., those valued $-1$ by all agents,
by \texttt{Round Robin} algorithm.
Then, either every agent receives exactly the same number of common chores
or there are two groups of agents with $k$ and $k-1$ such chores
for some $k \in \mathbbm{N}$.
In the former case, the allocation is envy-free, hence \DEF{(n-1)} as well.
In the latter case, if we give a dubious copy of some common chore to every agent
that receives $k-1$ common chores (and there are at most $n-1$ such),
we obtain an envy-free augmented allocation as well.
\end{proof}

We note that the algorithm described above
also guarantees that the output allocation satisfies \EFX{}.


For bivalued valuations,
we will build on the algorithms
introduced independently by
\citet{ebadian22:fairly} and
\citet{garg2022fair}
that for every such instance find an
\EF{1} and \PO{} allocation.
The algorithm relies on Fisher market equilibrium
and in order to use it,
we first show the following lemma.

\begin{lemma}
\label{lem:pef1}
Given an instance $\calI = \langle \calN, \calM, \calV \rangle$,
and a \pEF{1} Fisher market equilibrium $(A,p)$,
it holds that $A$ is a \DEF{(n-1)} and \fPO{} allocation.
\end{lemma}

\begin{proof}
Since an allocation in Fisher market equilibrium is always \fPO{} \citep{Mas-Colell95:Microeconomic},
it suffices to show that if $(A,p)$ is \pEF{1},
then allocation $A$ is $\DEF{(n-1)}$.

To this end, let us denote by $i_{\max}$ the agent with maximally priced bundle,
i.e., $i_{\max} \in \arg \max_{i \in \calN} p(A_i)$,
and by $c_{\max}$ a chore with the highest price,
i.e., $c_{\max} \in \arg \max_{c \in \calM} p(c)$.
Now, we construct a dubious allocation $A^D$ where $D$ contains $n-1$ copies of $c_{\max}$, denoted by $c_{\max}'$, which are allocated one each to agents $\calN \setminus \{i_{\max}\}$. We show that the obtained augmented allocation
$A^* = A \cup A^D$ is envy-free as follows.

First, observe that no agent envies $i_{\max}$, even in the real allocation $A$.
Indeed, every agent $i \in \calN$ receives only its minimum pain per buck chores,
therefore
\begin{align*}
    v_i(A_i) & = - p(A_i) \cdot \min_{c \in \calM} \textstyle\frac{|v_i(c)|}{p(c)}
             \ge p(A_{i_{\max}}) \cdot \max_{c \in \calM} \frac{v_i(c)}{p(c)} \\
             & \ge \textstyle\sum_{c \in A_{i_{\max}}} p(c) \cdot \textstyle\frac{v_i(c)}{p(c)}
             = v_i(A_{i_{\max}}).
\end{align*}

Hence, in $A^*$ also no agent envies $i_{\max}$.
Thus, let us take arbitrary $i,j \in \calN$
such that $i \neq i_{\max}$ and prove that
$j$ does not envy $i$.
Observe that
\begin{align*}
    v_j(A^*_i) & = v_j(A_i) + v_j(c_{\max}')\\
               &= \textstyle\sum_{c \in A_i} \left( p(c) \textstyle\frac{v_j(c)}{p(c)} \right)
                    + p(c_{\max}) \textstyle\frac{v_j(c_{\max}')}{p(c_{\max})} \\
               & \le - (p(A_i) + p(c_{\max})) \cdot \min_{c \in \calM} \textstyle\frac{|v_j(c)|}{p(c)}.
\end{align*}
Since $(A,p)$ is \pEF{1},
we get that $p(A_j) - p(c) \le p(A_i)$ for some $c \in \calM$.
Hence, $p(A_j) \le p(A_i) + p(c_{\max})$ and
\(
    v_j(A^*_i) \le - p(A_j) \cdot \min_{c \in \calM} \textstyle\frac{|v_j(c)|}{p(c)}
                = v_j(A_j)
                = v_j(A^*_j).
\)
In conclusion, $A^*$ is envy-free, so $A$ is \DEF{(n-1)}.
\end{proof}

\begin{figure*}[t!]
\centering
\begin{tabular}{ccc}
    \includegraphics[width=0.31\textwidth]{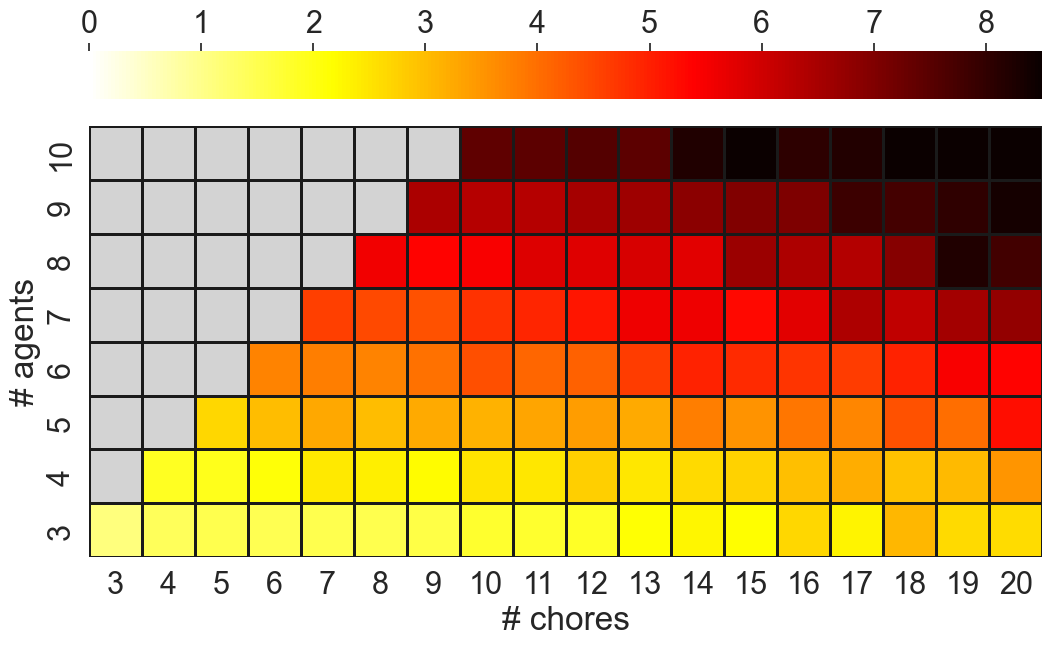}
     & 
    \includegraphics[width=0.31\textwidth]{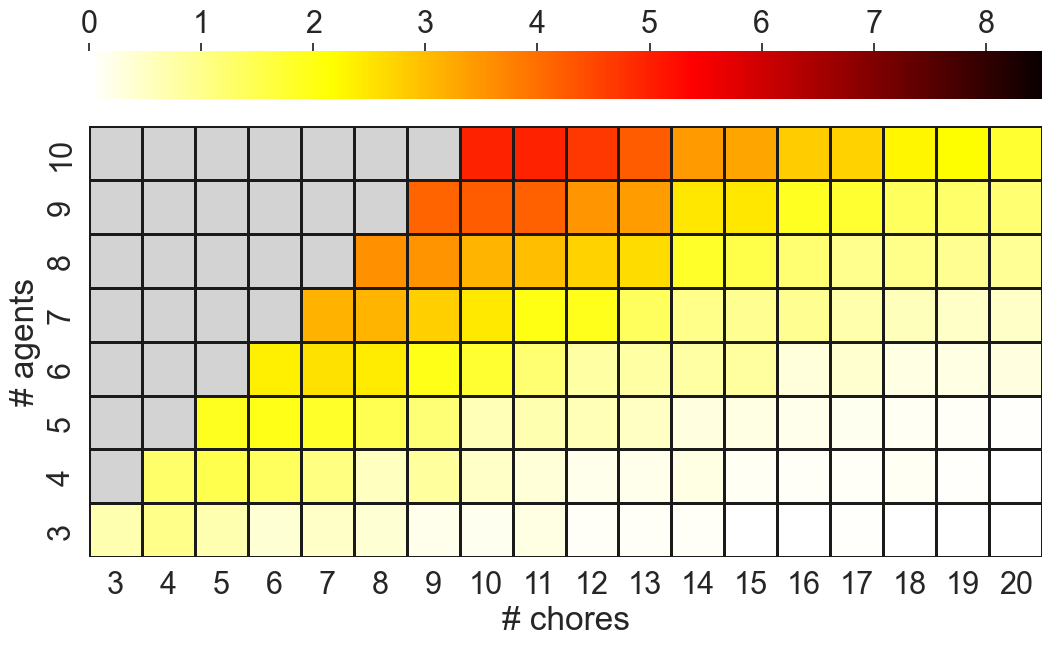}
     &
    \includegraphics[width=0.31\textwidth]{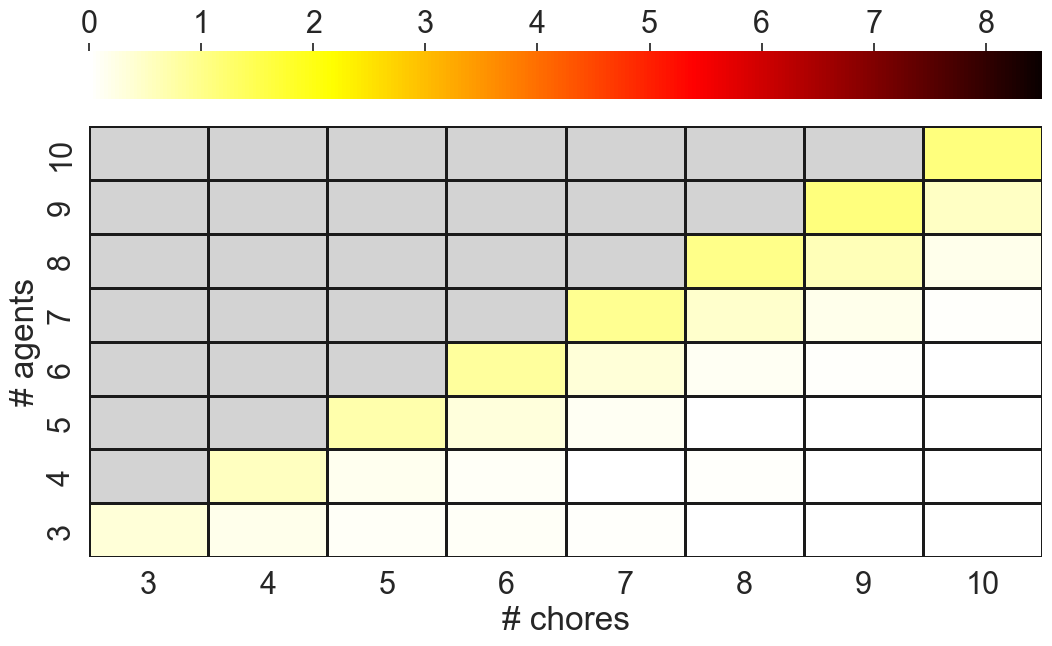}\\
    (a) \texttt{Envy Graph} & (b) \texttt{Round Robin} & (c) \texttt{PO} \\
\end{tabular}
\caption{Results for synthetic data for (a) \texttt{Envy Graph}, (b) \texttt{Round Robin}, and (c) \texttt{PO} algorithms. Grey tiles indicate no samples.}
\label{fig:exps}
\end{figure*}

\begin{theorem}
\label{thm:defn-1+po:bivalued}
Given an instance $\calI = \langle \calN, \calM, \calV \rangle$
with bivalued valuations,
there exists a \DEF{(n-1)} and \fPO{} allocation
and it can be computed in polynomial time.
\end{theorem}

\begin{proof}
\citet{garg2022fair} and
\citet{ebadian22:fairly} provided a polynomial -time algorithm that finds
a \pEF{1} Fisher market equilibrium
for every instance with bivalued valuations.
This, combined with \cref{lem:pef1}, yields the thesis.
\end{proof}


Finally, let us consider the case of two types of chores.

\begin{theorem}
\label{thm:defn-1+po:2types}
Given an instance $\calI = \langle \calN, \calM, \calV \rangle$
with two types of chores,
there exists a \DEF{(n-1)} and \fPO{} allocation
and it can be computed in polynomial time.
\end{theorem}

\begin{proof}
By definition, since the instance has two types of chores,
we can split $\calM$ into two sets $X$ and $Y$
such that for every agent $i \in \calN$ and chores $c, c'$ in one of the sets
we have $v_i(c) = v_i(c')$.
Thus, for every agent $i \in \calN$,
by $v^X_i$ and $v^Y_i$ let us denote
the agent's valuations of chores
from sets $X$ and $Y$ respectively.
\citet{aziz2023twotypes} provided a polynomial-time algorithm that
for every instance with two types of chores
finds an \fPO{} allocation $A$ and agent $i^* \in \calN$ such that:

\begin{itemize}
    \item[(1)] only agent $i^*$ may receive chores of two different types,
    \item[(2)] every agent $i \in \calN$ for which
    $v^X_i/v^X_{i^*} < v^Y_i/v^Y_{i^*}$
    receives only items of type $X$,
    \item[(3)] every agent $i \in \calN$ for which
    $v^X_i/v^X_{i^*} > v^Y_i/v^Y_{i^*}$
    receives only items of type $Y$, and
    \item[(4)] in the instance with identical valuations obtained from $\calI$
    by changing valuations of every agent to that of agent $i^*$,
    i.e., $\calI' = \langle \calN, \calM, v_{i^*} \rangle$,
    allocation $A$ is \EF{1}.
\end{itemize}

Thus, it suffices to show that the allocation satisfying these conditions
is also $\DEF{(n-1)}$. 
To this end, for every chore $c$ of type $X$ or $Y$
let us set $p(c) = v^X_{i^*}$ or $p(c) = v^Y_{i^*}$, respectively.
First, let us show that in such a case $(A,p)$
is a Fisher market equilibrium.
To this end, observe that
for every agent $i \in \calN$ with 
$v^X_i/v^X_{i^*} < v^Y_i/v^Y_{i^*}$
we have that $MPB_i$ consists of all chores of type $X$.
And by condition (2),
these are the only chores that agent $i$ receives.
Analogously, from condition (3) we get that
every agent $i \in \calN$ for which 
$v^X_i/v^X_{i^*} > v^Y_i/v^Y_{i^*}$
receives only its minimum pain per buck chores as well.
Finally, for agent $i^*$, we get that $MPB_{i^*} = \calM$.
Hence, every agent receives only
its minimal pain per buck chores,
which means that $(A,p)$
is indeed a Fisher market equilibrium.

Now, let us show that $(A,p)$ is \pEF{1}.
Fix arbitrary agents $i,j \in \calN$.
If $A_i \neq \emptyset$, then
by condition (4) there exists $c \in A_i$
such that $v_{i^*}(A_i \setminus \{c\}) \ge v_{i^*}(A_j)$.
Thus, we get $p(A_i \setminus \{c\}) \ge p(A_j)$,
which means that $(A,p)$ indeed is \pEF{1}.
Therefore, the thesis follows from \cref{lem:pef1}.
\end{proof}

\section{Experiments}

We experimentally investigate
the minimal number of dubious chores
needed to make an allocation \EF{}.
We generated a synthetic data set varying the number of agents $n$ from $3$ to $10$ and chores $m$ from $3$ to $10$ or $20$.
For each pair $(n,m)$, we generated $100$ instances with independent binary valuations such that for each $i \in \calN$ and $c \in \calM$ the valuation $v_{i}(c)$ is $-1$ with probability $0.7$ and $0$ with probability $0.3$.
To avoid the trivial \EF{} case, we assert $m \geq n$ and set the last chore in each instance to be valued at $-1$ for each agent.

Next, we computed allocations using three different algorithms: (a) \texttt{Envy Graph}, (b) \texttt{Round Robin}
, and (c) \texttt{PO}. 
The first two produce a single allocation per instance that satisfies \EF{1}.
The last algorithm searches through all \PO{} allocations and returns one that is $\DEF{k}$ for minimal $k$. 
\cref{fig:exps} shows heatmaps of average minimum number of dubious chores needed to make the output allocations \EF{}.

In \cref{fig:exps}a, we see that this minimum count for \texttt{Envy Graph} allocations increase with the number of agents or chores. For \texttt{Round Robin}, in \cref{fig:exps}b, we see that this count also increases with the number of agents but decreases as we have more chores. 
This may be explained by the fact that with a large number of chores, it is easier for agents to choose chores valued $0$ to them but possibly $-1$ to other agents.
With many chores assigned this way, the envy can be reduced using a smaller number of dubious chores.
Moreover, observe that the average number of dubious chores needed is generally smaller for allocation \texttt{Round Robin} than with \texttt{Envy Graph}.
The difference between these two algorithms matches our theoretical findings, as we have an $n-1$ bound for \texttt{Round Robin}
(\cref{thm:round-robin:defn-1}), but the same bound for \texttt{Envy Graph} does not hold (\cref{ex:ef1:defnn-1}).

Finally, when looking at the optimal \PO{}
solution in \cref{fig:exps}c,
we see that the required number of dubious chores
increases with the increase in the number of agents,
but drops sharply, when we increase the number of chores.
In fact, whenever the number of chores is not almost equal
to the number of agents, we can obtain envy-freeness
with one or zero dubious chores in most cases.
All in all, our experiment shows that in practice
the number of dubious chores needed
is much smaller than our theoretical guarantees indicate.

\section{Conclusion}

We have proposed a novel epistemic framework
for fair allocations of chores
through the introduction of dubious chores
and \DEF{k}, an approximation
of envy-freeness which details the trade-off between fairness and transparency of an allocation.
Although finding \DEF{k} allocations with small $k$
is computationally hard, 
we have provided several guarantees
for the existence of such allocations
with and without 
Pareto optimality.
Our experimental results suggest that 
the number of dubious chores required to make an allocation free of envy
is lower in practice than our theoretical guarantees indicate.

Some of the problems considered in this paper
are still not fully resolved.
In particular,
the existence of
a \DEF{(n-1)} and \PO{} chore allocation
in every instance 
for agents with additive valuations seems possible.
Proving this can be a valuable step to
resolve the open existence problem of
\EF{1} and \PO{} allocations for chores.

\section*{Acknowledgements}
HH acknowledges support from NSF IIS grants \#2144413 and \#2107173. LX acknowledges NSF \#2007994 and  \#2106983. TW was partially supported by EPSRC under grant EP/X038548/.



\bibliographystyle{plainnat}
\bibliography{arxiv_main}


\end{document}

%% file: macros.tex
\usepackage{latexsym}
\usepackage{amssymb}
\usepackage{amsmath}
\usepackage{amsthm}
\usepackage{booktabs}
\usepackage{enumitem}
\usepackage{color}

\usepackage{etoolbox} 
\newcommand{\tightdots}{\cdot\hspace{-2.2pt}\cdot\hspace{-2.2pt}\cdot}	
\usepackage{graphicx} 
\usepackage[capitalise,noabbrev]{cleveref}
\usepackage{bbm}
\usepackage{natbib}
\usepackage{thmtools}

\usepackage{tikz}

\newcommand*\circled[1]{\tikz[baseline=(char.base)]{\node[shape=circle,draw,inner sep=2pt] (char) {#1};}}
\newcommand*\dashedcirc[1]{\tikz[baseline=(char.base)]{\node[shape=circle,draw,densely dotted,inner sep=2pt] (char) {#1};}}

\newcommand{\calI}{\mathcal{I}}
\newcommand{\calV}{\mathcal{V}}
\newcommand{\calN}{\mathcal{N}}
\newcommand{\calM}{\mathcal{M}}

\newcommand{\EF}[1]{\ifstrempty{#1}{\textrm{\textup{EF}}}{\textrm{\textup{EF{$#1$}}}}}
\newcommand{\EFX}[1]{\ifstrempty{#1}{\textrm{\textup{EFX}}}{\textrm{\textup{EFX{$#1$}}}}}
\newcommand{\sEF}[1]{\textrm{\textup{sEF{$#1$}}}}
\newcommand{\DEF}[1]{\ifstrempty{#1}{\textrm{\textup{DEF}}}{\textrm{\textup{DEF-{$#1$}}}}}
\newcommand{\pEF}[1]{\ifstrempty{#1}{\textrm{\textup{pEF}}}{\textrm{\textup{pEF{$#1$}}}}}

\newcommand{\PO}[1]{\ifstrempty{#1}{\textrm{\textup{PO}}}{\textrm{\textup{PO{$#1$}}}}}
\newcommand{\fPO}[1]{\ifstrempty{#1}{\textrm{\textup{fPO}}}{\textrm{\textup{fPO{$#1$}}}}}

\usepackage[capitalise,noabbrev]{cleveref}
\Crefname{remark}{Remark}{Remarks}
\Crefname{rmk}{Remark}{Remarks}
\Crefname{dfn}{Definition}{Definitions}
\Crefname{thm}{Theorem}{Theorems}
\Crefname{cor}{Corollary}{Corollaries}
\Crefname{lem}{Lemma}{Lemmas}
\Crefname{examplex}{Example}{Examples}
\Crefname{prop}{Proposition}{Propositions}

\newtheorem{example}{Example}
\newtheorem{theorem}{Theorem}
\newtheorem*{remark}{Remark}

\newtheorem{definition}{Definition}
\newtheorem{coro}{Corollary}
\newtheorem{lemma}{Lemma}
\newtheorem{proposition}{Proposition}

\newcount\Comments  
\newcommand{\kibitz}[2]{\ifnum\Comments=1{\color{#1}{#2}}\fi}
